\documentclass[10pt,a4paper]{article}

\usepackage{arxiv}

\usepackage[utf8]{inputenc}
\usepackage[T1]{fontenc}
\usepackage{amsmath,amssymb,amsfonts,amsthm}
\usepackage{graphicx}
\usepackage{booktabs}
\usepackage{tikz}
\usetikzlibrary{shapes,arrows,positioning,calc,decorations.pathmorphing,patterns,arrows.meta,decorations.markings,shapes.geometric}
\usepackage{hyperref}
\usepackage{cleveref}
\usepackage{natbib}
\usepackage{doi}
\usepackage{textcomp}

\graphicspath{{figures/}}

\newtheorem{theorem}{Theorem}
\newtheorem{lemma}[theorem]{Lemma}
\newtheorem{proposition}[theorem]{Proposition}
\newtheorem{corollary}[theorem]{Corollary}
\newtheorem{definition}{Definition}

\newtheorem{assumption}{Assumption}
\newtheorem{remark}{Remark}

\newcommand{\R}{\mathbb{R}}
\newcommand{\N}{\mathbb{N}}
\newcommand{\Z}{\mathbb{Z}}
\newcommand{\dd}{\mathrm{d}}
\newcommand{\Sing}{\mathcal{S}}
\newcommand{\Val}{\mathcal{V}}
\newcommand{\Unc}{\mathcal{U}}
\newcommand{\Obs}{\mathcal{O}}
\newcommand{\Res}{\mathcal{R}}
\newcommand{\Xadm}{\mathcal{X}_{\mathrm{adm}}}

\title{On the Minimum Number of Control Laws for Nonlinear Systems\\with Input--Output Linearisation Singularities}

\author{
  Nikolaos D.~Tantaroudas\thanks{Corresponding author. Senior Researcher, ICCS.} \\
  Institute of Communications and Computer Systems (ICCS)\\
  National Technical University of Athens\\
  9 Iroon Politechniou Street, Zografou, Athens 15773, Greece\\
  \texttt{nikolaos.tantaroudas@iccs.gr} \\
}

\date{}

\begin{document}

\maketitle

\begin{abstract}
This paper studies how many distinct feedback-linearising control laws are required to control a single-input single-output nonlinear system whose exact input--output linearising law is singular on a nonempty subset of the state space. The study examines whether exact and approximate (Hauser--Sastry--Kokotovi\'{c}) linearising laws covers a region if every state of the region lies in the validity domain of at least one law. Using transversality arguments from differential topology we prove that, for a family of $p$ laws whose singular hypersurfaces are in general position, the uncovered set is a finite union of embedded submanifolds of codimension exactly $p$. Consequently $n+1$ laws always suffice for complete coverage of an $n$-dimensional state space, two laws already reduce the uncovered set to a thin set that generic trajectories never meet, and complete coverage with fewer than $n+1$ laws is a question about specific intersections rather than about the number of components of the singularity manifold. The theory is applied exhaustively to the ball-and-beam system, whose exact linearising law is singular on the union of two transverse hypersurfaces $\{r=0\}\cup\{\dot\theta=0\}$. We construct three explicit laws, the exact law and the two canonical Hauser--Sastry--Kokotovi\'{c} approximations, compute their validity domains in closed form and prove that (i) every periodic tracking task forces the state through the singular set, so that the exact law can never be used alone; (ii) a single approximate law covers the entire physically admissible region $|\theta|<\pi/2$, and any two of the three laws cover $\R^4$ up to a codimension-two set; (iii) the three laws together cover $\R^4$ minus an explicit obstructed set on which the output is provably not linearisable by any static feedback, because the ball motion there is invariant under reversal of the input. Numerical simulations confirm each statement.
\end{abstract}


%======================================================================
\section{Introduction}
\label{sec:introduction}
%======================================================================

Feedback linearisation stands as one of the most powerful and elegant techniques in nonlinear control theory, enabling the transformation of complex nonlinear systems into equivalent linear systems through judicious coordinate changes and state feedback~\citep{Isidori1995,Khalil2002,NijmeijervanderSchaft1990}. The central idea, developed extensively by Isidori, Krener and their collaborators during the 1980s, exploits the differential geometric structure inherent in nonlinear systems to cancel nonlinearities exactly, thereby yielding a linear input--output relationship amenable to the full arsenal of classical linear design techniques~\citep{IsidoriKrener1982,Jakubczyk1980,Hunt1983}.

However, the conditions needed for exact input--output linearisation are highly restrictive and frequently fail to hold globally across the entire state space. The seminal work of \citet{Jakubczyk1980} established necessary and sufficient conditions for exact state linearisation through the involutivity of certain distributions constructed from the system's vector fields. When these geometric conditions fail, alternative approaches become necessary. \citet{Krener1984} showed that approximate linearisation may be achieved even when exact conditions fail, while \citet{ByrnesIsidori1988} characterised the local stabilisation properties of minimum-phase nonlinear systems exhibiting such failures.

A particularly important and challenging class of problems arises when the \emph{relative degree} of a system becomes undefined on certain submanifolds of the state space, creating \emph{singularities} in the feedback linearising control law. At such singular points the control coefficient that multiplies the input vanishes, rendering the standard linearising feedback undefined. This phenomenon is not merely a mathematical curiosity but arises naturally in numerous practical systems of engineering interest.

\subsection{Literature review and motivation}

The geometric approach to nonlinear control has its intellectual roots in the pioneering work of \citet{Brockett1978}, who first considered certain forms of feedback in addition to pure state transformations. \citet{Hunt1983} provided fundamental results on global transformations of nonlinear systems, while the comprehensive theory was unified and extended by \citet{Isidori1995} in his seminal treatise, which remains the standard reference in the field. The textbooks by \citet{NijmeijervanderSchaft1990}, \citet{Khalil2002}, \citet{Sastry1999}, \citet{SlotineLi1991} and \citet{Marino1995} provide accessible introductions to these geometric methods and their applications.

The problem of systems with singularities has been addressed from various perspectives. \citet{Hauser1992} introduced approximate input--output linearisation for systems whose relative degree is not well defined: a nearby system with well-defined relative degree is constructed by neglecting terms that are small in the operating region, and the linearising law of the approximating system is applied to the true one. \citet{Tomlin1998} proposed switching between the exact linearising law away from the singularities and the approximate law of \citet{Hauser1992} close to them, and identified the behaviour of the internal (zero) dynamics at the switching surface as the decisive issue for the applicability of such a switched law. \citet{Liberzon2003} developed a comprehensive theory of switching in systems and control, providing fundamental stability criteria for switched systems based on common Lyapunov functions and dwell-time conditions. The concept of \emph{patchy vector fields} introduced by \citet{Ancona1999} guarantees the existence of finite controller families for asymptotic stabilisation of nonlinear systems, though without quantifying the number required. Sontag's influential work on input-to-state stability~\citep{Sontag1989,Sontag1995} provides additional analytical tools for assessing the robustness of such switched control architectures.

The ball-and-beam system provides an ideal canonical case study for investigating feedback linearisation singularities. This classical benchmark, found in control laboratories worldwide, exhibits all the essential features of systems with ill-defined relative degree while remaining sufficiently simple for complete analytical treatment. The foundational work of \citet{Hauser1992} demonstrated rigorously that the ball-and-beam fails the involutivity conditions for full state linearisation, yet showed that approximate input--output linearisation can nonetheless be achieved through careful modification of the system's vector fields. Their work established the theoretical foundation for the approximate linearisation techniques that we analyse here.

The mathematical foundations underlying our proofs rely on results from differential topology, particularly transversality theory as developed in the classic texts by \citet{Hirsch1976} and \citet{GuilleminPollack1974}. The Implicit Function Theorem, which plays a central role in our arguments regarding the local structure of singularity manifolds, is treated in \citet{KrantzParks2002}. We also draw upon the theory of smooth manifolds as presented in \citet{Lee2012}.

Despite this extensive body of literature, a basic question has not been answered in geometric terms: \emph{given a nonlinear system whose exact linearising law is singular on a set $\Sing$ of known geometric structure, how many distinct linearising control laws, exact or approximate, are needed so that every state of a prescribed region is covered by at least one well-defined law, and which states can never be covered?} While practical designs using several controllers have appeared~\citep{Hauser1992,Tomlin1998}, and while patchy feedback theory~\citep{Ancona1999} guarantees the existence of finite controller families, the relation between the number of laws, the dimension of the state space and the codimension of the set left uncovered has not been made explicit.

\subsection{Contributions}

The present work makes the following contributions:

\begin{enumerate}
\item \textbf{A codimension theorem for controller coverage.} For a family of $p$ linearising laws whose singular hypersurfaces are in general position, the set of states covered by none of the laws is a finite union of embedded submanifolds of codimension exactly $p$ (\Cref{thm:codim}). Hence $n+1$ laws in general position always give complete coverage of an $n$-dimensional state space, two laws already leave only a thin set that generic trajectories never meet (\Cref{lem:thin}), and complete coverage with $p\le n$ laws holds precisely when finitely many explicitly computable intersections are empty.

\item \textbf{Complete analysis of the ball-and-beam system.} We provide a self-contained Lagrangian derivation of the ball-and-beam dynamics, identify the singularity manifold $\Sing=\{r=0\}\cup\{\dot\theta=0\}$ as a regular two-component manifold in general position, construct three explicit control laws (the exact law and the two canonical approximations of \citet{Hauser1992}) and compute their validity domains in closed form (\Cref{thm:bb}).

\item \textbf{An obstruction result.} We exhibit an explicit two-dimensional set $\Obs\subset\R^4$ (beam vertical and momentarily at rest, or beam vertical with the ball at rest at the pivot) on which no static feedback can linearise the output of the ball-and-beam at any order, because the ball motion from such states is invariant under reversal of the input (\Cref{prop:obstruction,prop:symmetry}). The three laws cover exactly $\R^4\setminus\Obs$, so the family is maximal.

\item \textbf{Necessity of a second law.} We prove that every trajectory that tracks a non-constant periodic reference with bounded input and bounded beam angle crosses the singular component $\{\dot\theta=0\}$ infinitely often (\Cref{prop:crossings}); the exact law alone can therefore never execute a periodic tracking task, whereas one approximate law covers the whole admissible region $|\theta|<\pi/2$.

\item \textbf{Numerical validation.} Simulations confirm the finite-time loss of the exact law, the coverage of the admissible region by a single approximate law, the codimension of the residual sets, the scaling of the approximation error with the reference speed and the input-reversal symmetry on the obstructed set.
\end{enumerate}

\subsection{Paper organisation}

The remainder of this paper is organised as follows. \Cref{sec:preliminaries} reviews the required concepts from nonlinear control theory and differential topology. \Cref{sec:ballbeam} derives the equations of motion for the ball-and-beam system using Lagrangian mechanics. \Cref{sec:singularity} performs the input--output linearisation procedure, analyses the singularity manifold and establishes the obstruction result. \Cref{sec:controllers} develops the three control laws and computes their validity domains. \Cref{sec:theorem} presents the general coverage theory and its consequences for the ball-and-beam. \Cref{sec:architecture} describes the closed-loop architecture. \Cref{sec:simulations} provides numerical validation. \Cref{sec:discussion} discusses implications and limitations. \Cref{sec:conclusion} concludes.

%======================================================================
\section{Mathematical Preliminaries}
\label{sec:preliminaries}
%======================================================================

This section establishes the mathematical framework and notation used throughout the paper.

\subsection{Nonlinear control systems}

We consider single-input single-output (SISO) nonlinear control systems in control-affine form:
\begin{equation}
\label{eq:system}
\dot{x} = f(x) + g(x)u, \qquad y = h(x),
\end{equation}
where $x \in \mathcal{X}$ is the state, $\mathcal{X}\subseteq\R^n$ is open, $u \in \R$ is the scalar control input, $y \in \R$ is the scalar measured output and $f: \mathcal{X} \to \R^n$, $g: \mathcal{X} \to \R^n$, $h: \mathcal{X} \to \R$ are smooth (i.e.\ $C^\infty$) mappings. Throughout, the gravitational acceleration is denoted $G$, so that the letter $g$ is reserved for the input vector field.

\subsection{Lie derivatives and Lie brackets}

Differential geometry provides the natural language for analysing nonlinear control systems~\citep{Isidori1995,NijmeijervanderSchaft1990}.

\begin{definition}[Lie Derivative]
For a smooth function $\varphi: \mathcal{X} \to \R$ and a smooth vector field $v: \mathcal{X} \to \R^n$, the \emph{Lie derivative} of $\varphi$ along $v$ is
\begin{equation}
L_v\varphi(x) := \nabla\varphi(x) \cdot v(x) = \sum_{i=1}^{n} \frac{\partial \varphi}{\partial x_i}(x) \, v_i(x).
\end{equation}
\end{definition}

Higher-order Lie derivatives are defined recursively: $L_v^0\varphi := \varphi$ and $L_v^j\varphi := L_v(L_v^{j-1}\varphi)$ for $j \geq 1$.

\begin{definition}[Lie Bracket]
The \emph{Lie bracket} of two smooth vector fields $f, g: \mathcal{X} \to \R^n$ is
\begin{equation}
[f,g](x) := \frac{\partial g}{\partial x}(x) \, f(x) - \frac{\partial f}{\partial x}(x) \, g(x).
\end{equation}
\end{definition}

We write $\mathrm{ad}_f^0 g := g$ and $\mathrm{ad}_f^j g := [f, \mathrm{ad}_f^{j-1} g]$ for $j \geq 1$.

\subsection{Relative degree and exact input--output linearisation}

The concept of relative degree is fundamental to input--output linearisation~\citep{Isidori1995}.

\begin{definition}[Relative Degree]
\label{def:relative_degree}
System~\eqref{eq:system} has \emph{relative degree} $\gamma \in \N$ at $x^0 \in \mathcal{X}$ if:
\begin{enumerate}
\item[(i)] $L_gL_f^jh(x) = 0$ for all $j \in \{0, 1, \ldots, \gamma-2\}$ and all $x$ in some open neighbourhood of $x^0$;
\item[(ii)] $L_gL_f^{\gamma-1}h(x^0) \neq 0$.
\end{enumerate}
\end{definition}

When the relative degree $\gamma$ is well defined at $x^0$, repeated differentiation of $y = h(x)$ yields
\begin{align}
y &= h(x), \notag \\
\dot{y} &= L_fh(x), \notag \\
&\;\;\vdots \notag \\
y^{(\gamma-1)} &= L_f^{\gamma-1}h(x), \\
y^{(\gamma)} &= L_f^\gamma h(x) + L_gL_f^{\gamma-1}h(x) \cdot u. \notag
\end{align}
Defining
\begin{equation}
\label{eq:ab_functions}
b(x) := L_f^\gamma h(x), \qquad a(x) := L_gL_f^{\gamma-1}h(x),
\end{equation}
we obtain $y^{(\gamma)} = b(x) + a(x)\,u$. The function $a(x)$ is called the \emph{control coefficient}. When $a(x) \neq 0$, the feedback
\begin{equation}
\label{eq:linearizing_control}
u = \frac{1}{a(x)}\bigl[-b(x) + v\bigr]
\end{equation}
yields $y^{(\gamma)} = v$, where $v$ is a virtual input designed using linear techniques.

\begin{definition}[Linearising law of order $\ell$ at a point]
\label{def:law_order}
A smooth static feedback $u = \kappa(x,v)$ is a \emph{linearising law of order $\ell$ at $x^0$} for system~\eqref{eq:system} if, for every $v\in\R$, the closed-loop output satisfies $y^{(\ell)} = v$ at $x^0$. The exact law~\eqref{eq:linearizing_control} is a linearising law of order $\gamma$ at every point where $a(x)\neq 0$.
\end{definition}

\subsection{Singularities and the singularity manifold}

The linearising control~\eqref{eq:linearizing_control} is undefined when $a(x) = 0$.

\begin{definition}[Singularity at the $m$th derivative]
\label{def:m_singularity}
System~\eqref{eq:system} exhibits a \emph{singularity at the $m$th output derivative} ($m \le n$) if $L_gL_f^jh \equiv 0$ on $\mathcal{X}$ for all $j \in \{0,\ldots,m-2\}$ and the control coefficient $a := L_gL_f^{m-1}h$ vanishes on a nonempty proper closed subset
\begin{equation}
\Sing := \{x \in \mathcal{X} : a(x) = 0\},
\end{equation}
called the \emph{singularity manifold}. The \emph{exact law} $u_1(x,v) := (v - L_f^m h(x))/a(x)$ is then a linearising law of order $m$ on the \emph{regular region} $\Val_1 := \mathcal{X}\setminus\Sing$.
\end{definition}

\begin{definition}[Regular hypersurface-general position]
\label{def:general_position}
A \emph{regular hypersurface} of $\mathcal{X}$ is a set $N = \{x\in\mathcal{X} : \rho(x) = 0\}$ with $\rho$ smooth and $\dd\rho(x)\neq 0$ for all $x\in N$; by the Implicit Function Theorem $N$ is a closed embedded submanifold of codimension one. A finite family $\{N_1,\ldots,N_q\}$ of regular hypersurfaces, $N_j = \{\rho_j = 0\}$, is \emph{in general position} if, for every $x\in\bigcup_j N_j$, the covectors $\{\dd\rho_j(x) : \rho_j(x) = 0\}$ are linearly independent.
\end{definition}

\begin{definition}[Regular $k$-component singularity manifold]
\label{def:k_component}
The singularity manifold $\Sing$ is a \emph{regular $k$-component singularity manifold} if there exist smooth functions $\varphi_1,\ldots,\varphi_k : \mathcal{X}\to\R$ such that
\begin{enumerate}
\item[(i)] $a(x) = c(x)\prod_{i=1}^k\varphi_i(x)$ for some smooth $c$ with $c(x)\neq 0$ for all $x\in\mathcal{X}$;
\item[(ii)] each $\Sing_i := \{x : \varphi_i(x) = 0\}$ is a nonempty connected regular hypersurface;
\item[(iii)] the family $\{\Sing_1,\ldots,\Sing_k\}$ is in general position.
\end{enumerate}
Then $\Sing = \bigcup_{i=1}^k\Sing_i$, and the $\Sing_i$ are called the \emph{components} of $\Sing$.
\end{definition}

\begin{remark}
\label{rem:k_intrinsic}
The number $k$ is intrinsic to the set $\Sing$. Indeed, at a point where two components meet, general position implies that the two branches of $\Sing$ have distinct tangent hyperplanes, and a connected embedded hypersurface contained in $\Sing$ must continue through such a point along the branch to which it is tangent. Hence any two decompositions of $\Sing$ into connected regular hypersurfaces in general position coincide up to the order of the components. The factored form $a = c\prod\varphi_i$ reflects the structure encountered in mechanical systems, where the control coefficient is a product of kinematic terms.
\end{remark}

\subsection{Approximate input--output linearisation}

When $\Sing\neq\emptyset$, \citet{Hauser1992} propose to linearise a nearby system instead of the true one. Their construction can be phrased directly in terms of the output-derivative chain.

\begin{definition}[Approximate linearising chain of order $\ell$]
\label{def:chain}
An \emph{approximate linearising chain of order $\ell$} is a list of smooth functions $\phi_1 := h, \phi_2, \ldots, \phi_\ell$ on $\mathcal{X}$ with $\ell\le n$. Its \emph{coefficient} is $\alpha := L_g\phi_\ell$, its \emph{validity domain} is the open set
\begin{equation}
\Val := \{x\in\mathcal{X} : \alpha(x)\neq 0\},
\end{equation}
and the set $\Sing' := \{x\in\mathcal{X} : \alpha(x) = 0\}$ is its \emph{singular set}. The functions
\begin{equation}
\psi_j := L_f\phi_j - \phi_{j+1}, \qquad \chi_j := L_g\phi_j, \qquad j = 1,\ldots,\ell-1,
\end{equation}
are the \emph{drift defects} and \emph{input defects} of the chain. The associated \emph{approximate law} is
\begin{equation}
\label{eq:approx_law}
u(x,v) = \frac{1}{\alpha(x)}\bigl[-L_f\phi_\ell(x) + v\bigr], \qquad x\in\Val .
\end{equation}
The chain is \emph{consistent} on a set $Z\subseteq\mathcal{X}$ if all its defects vanish on $Z$.
\end{definition}

Along the true system~\eqref{eq:system} in closed loop with~\eqref{eq:approx_law} one has
\begin{equation}
\label{eq:chain_closed_loop}
\dot\phi_j = \phi_{j+1} + \psi_j + \chi_j u \quad (j<\ell), \qquad \dot\phi_\ell = v,
\end{equation}
i.e.\ the closed loop is an integrator chain perturbed by the defects. If all defects vanish identically, $\phi_j = L_f^{j-1}h$ and~\eqref{eq:approx_law} is the exact law of order $\ell$. The tracking-error bounds of \citet{Hauser1992} state that the error incurred by using~\eqref{eq:approx_law} on the true system is bounded in terms of the size of the defects along the desired trajectory. The approximation is therefore useful precisely where the defects are small, which in our setting means close to the set $Z$ on which the chain is consistent. The restriction $\ell\le n$ reflects the construction of \citet{Hauser1992}. The chain consists of the output derivatives of an $n$-dimensional approximating system with well-defined relative degree $\ell$, and the relative degree of an $n$-dimensional system never exceeds $n$ \citep[Lemma~4.1.1]{Isidori1995}. Every exact or approximate linearising law considered in this paper therefore has order at most $n$.

\subsection{Elements of differential topology}

Our proofs rely on the following standard facts~\citep{Hirsch1976,GuilleminPollack1974,Lee2012}.

\begin{theorem}[Regular level sets {\citep{Lee2012,KrantzParks2002}}]
\label{thm:implicit}
Let $F: \mathcal{X} \to \R^m$ be smooth with $m \leq n$. If $\dd F(x)$ has full row rank at every point of $F^{-1}(0)$, then $F^{-1}(0)$ is a closed embedded submanifold of $\mathcal{X}$ of codimension $m$ (possibly empty).
\end{theorem}

\begin{lemma}[Intersections in general position]
\label{lem:intersections}
Let $\{N_1,\ldots,N_q\}$ be regular hypersurfaces in general position, $N_j = \{\rho_j = 0\}$. For every nonempty $J\subseteq\{1,\ldots,q\}$ the set $N_J := \bigcap_{j\in J}N_j$ is either empty or a closed embedded submanifold of codimension $|J|$. If $|J| > n$ then $N_J = \emptyset$.
\end{lemma}

\begin{proof}
$N_J = F_J^{-1}(0)$ with $F_J := (\rho_j)_{j\in J}$. At $x\in N_J$ all $\rho_j$, $j\in J$, vanish, so by general position the rows $\dd\rho_j(x)$, $j\in J$, of $\dd F_J(x)$ are linearly independent; \Cref{thm:implicit} applies. If $|J|>n$, more than $n$ linearly independent covectors would have to exist at a point of $N_J$, which is impossible; hence $N_J=\emptyset$.
\end{proof}

\begin{lemma}[Thin sets are not reached from generic initial states]
\label{lem:thin}
Let $F$ be a smooth vector field on an open set $\mathcal{W}\subseteq\R^n$ with flow $\Phi_t$, and let $N\subset\mathcal{W}$ be a finite union of embedded submanifolds of dimension at most $n-2$. Then the set
\begin{equation}
\mathcal{B}(N) := \{x_0\in\mathcal{W} : \Phi_t(x_0)\in N \text{ for some } t\ge 0 \text{ in the maximal interval of existence}\}
\end{equation}
has Lebesgue measure zero.
\end{lemma}

\begin{proof}
Write $N = \bigcup_i N_i$ with each $N_i$ an embedded submanifold of dimension $\le n-2$. The set $D_i := \{(z,t)\in N_i\times\R : \Phi_{-t}(z) \text{ is defined}\}$ is open in $N_i\times\R$ and $\Psi_i(z,t) := \Phi_{-t}(z)$ is smooth on it. Every $x_0\in\mathcal{B}(N)$ satisfies $x_0 = \Psi_i(z,t)$ for some $i$ and some $(z,t)\in D_i$ with $t\ge 0$, so $\mathcal{B}(N)\subseteq\bigcup_i\Psi_i(D_i)$. Each $D_i$ is a smooth manifold of dimension at most $n-1<n$, and the smooth image of a manifold of dimension less than $n$ has Lebesgue measure zero in $\R^n$ \citep[Cor.~6.11]{Lee2012}. A finite union of null sets is null.
\end{proof}

%======================================================================
\section{The Ball-and-Beam System: Lagrangian Derivation}
\label{sec:ballbeam}
%======================================================================

The ball-and-beam system serves as the primary example for illustrating and validating our theoretical results.

\subsection{Physical description and coordinate system}

Consider a rigid beam rotating freely in a vertical plane about a pivot at its centre, as depicted in \Cref{fig:ballbeam_schematic}. A torque $\tau$ applied at the pivot controls the beam angle. A solid spherical ball of mass $M$ and radius $R$ rolls without slipping along the beam surface. The system has two mechanical degrees of freedom.

\begin{figure}[htbp]
\centering
\begin{tikzpicture}[scale=1.1]
    \draw[->,thick,gray] (-0.5,0) -- (4.5,0) node[right] {$x$};
    \draw[->,thick,gray] (0,-0.5) -- (0,3) node[above] {$y$};
    \fill[black] (2,1.5) circle (3pt);
    \node[below left] at (2,1.5) {pivot};
    \draw[very thick, brown!60!black, rotate around={25:(2,1.5)}] (-0.3,1.5) -- (4.3,1.5);
    \draw[thick, brown!60!black, rotate around={25:(2,1.5)}] (-0.3,1.35) -- (4.3,1.35);
    \fill[blue!40] (3.2,2.35) circle (0.25);
    \draw[thick,blue!70!black] (3.2,2.35) circle (0.25);
    \fill[red] (3.2,2.35) circle (1.5pt);
    \draw[<->,thick,red] (2,1.65) -- (3.2,2.2) node[midway,above left] {$r$};
    \draw[thick,dashed] (2,1.5) -- (3.5,1.5);
    \draw[thick,->] (2.8,1.5) arc (0:25:0.8);
    \node at (3.1,1.65) {$\theta$};
    \draw[thick,->,green!50!black] (2,1.5) ++(0.3,0.3) arc (45:135:0.42);
    \node[green!50!black] at (1.4,2.1) {$\tau$};
    \draw[thick,->,purple] (3.2,2.35) -- (3.2,1.6) node[right] {$MG$};
    \draw[thick] (1.5,0.8) -- (2.5,0.8);
    \draw[thick] (2,0.8) -- (2,1.5);
    \fill[pattern=north east lines] (1.5,0.6) rectangle (2.5,0.8);
    \node[blue!70!black] at (3.2,2.8) {ball (mass $M$)};
    \node[brown!60!black] at (0.5,1.8) {beam};
    \draw[<->,thin] (3.2,2.35) -- (3.45,2.35) node[right,font=\small] {$R$};
\end{tikzpicture}
\caption{Schematic of the ball-and-beam system. The beam rotates about a fixed pivot with angle $\theta$ from horizontal. The ball position $r$ is the signed distance from pivot to ball centre along the beam. The torque $\tau$ is the control input.}
\label{fig:ballbeam_schematic}
\end{figure}

We parameterise the configuration using: $r$, the signed distance from pivot to ball centre along the beam (positive to the right), and $\theta$, the beam angle measured counterclockwise from horizontal.

The physical parameters are: $M$ (ball mass), $R$ (ball radius), $J_b$ (ball moment of inertia about its centre), $J$ (beam moment of inertia about pivot) and $G$ (gravitational acceleration). For a solid sphere, $J_b = \frac{2}{5}MR^2$.

\subsection{Kinetic and potential energy}

The ball's centre-of-mass position in the inertial frame is
$\mathbf{p} = r\cos\theta\,\hat{\mathbf{e}}_x + r\sin\theta\,\hat{\mathbf{e}}_y$,
giving velocity magnitude $|\dot{\mathbf{p}}|^2 = \dot{r}^2 + r^2\dot{\theta}^2$.

The total kinetic energy comprises translational energy of the ball, rotational energy of the ball about its own centre and rotational energy of the beam. Since the ball rolls without slipping, its angular velocity in the inertial frame is $\omega_{\text{ball}} = \dot{r}/R + \dot{\theta}$. Thus:
\begin{equation}
\label{eq:kinetic_energy}
T = \frac{1}{2}M|\dot{\mathbf{p}}|^2 + \frac{1}{2}J_b\!\left(\frac{\dot{r}}{R} + \dot{\theta}\right)^{\!2} + \frac{1}{2}J\dot{\theta}^2.
\end{equation}

Adopting the standard approximation of \citet{Hauser1992} (the cross term $J_b\dot{r}\dot{\theta}/R$ is negligible for small ball radius):
\begin{equation}
\label{eq:kinetic_approx}
T \approx \frac{1}{2}\!\left(M + \frac{J_b}{R^2}\right)\dot{r}^2 + \frac{1}{2}(Mr^2 + J + J_b)\dot{\theta}^2.
\end{equation}

The gravitational potential energy is $V = MGr\sin\theta$.

\subsection{Euler--Lagrange equations}

The Lagrangian $\mathcal{L} = T - V$ gives:
\begin{equation}
\label{eq:lagrangian}
\mathcal{L} = \frac{1}{2}\!\left(M + \frac{J_b}{R^2}\right)\dot{r}^2 + \frac{1}{2}(Mr^2 + J + J_b)\dot{\theta}^2 - MGr\sin\theta.
\end{equation}

Applying the Euler--Lagrange equations yields:
\begin{align}
\left(M + \frac{J_b}{R^2}\right)\ddot{r} - Mr\dot{\theta}^2 + MG\sin\theta &= 0, \label{eq:eom_r} \\
(Mr^2 + J + J_b)\ddot{\theta} + 2Mr\dot{r}\dot{\theta} + MGr\cos\theta &= \tau, \label{eq:eom_theta}
\end{align}
which match those of \citet{Hauser1992}.

\subsection{State-space representation}

Define $x = (x_1, x_2, x_3, x_4)^\top := (r, \dot{r}, \theta, \dot{\theta})^\top$ and
\begin{equation}
B := \frac{M}{M + J_b/R^2} = \frac{5}{7} \quad \text{(solid sphere)}.
\end{equation}

From~\eqref{eq:eom_r}: $\ddot{r} = B(x_1 x_4^2 - G\sin x_3)$.

Applying the preliminary feedback
\begin{equation}
\label{eq:preliminary_feedback}
\tau = 2Mr\dot{r}\dot{\theta} + MGr\cos\theta + (Mr^2 + J + J_b)u
\end{equation}
reduces~\eqref{eq:eom_theta} to $\ddot{\theta} = u$. The state-space form is then:
\begin{equation}
\label{eq:statespace}
\begin{aligned}
\dot{x}_1 &= x_2, \\
\dot{x}_2 &= B(x_1 x_4^2 - G\sin x_3), \\
\dot{x}_3 &= x_4, \\
\dot{x}_4 &= u,
\end{aligned}
\end{equation}
with output $y = h(x) = x_1$ (ball position) and state space $\mathcal{X} = \R^4$. In the form~\eqref{eq:system}:
\begin{equation}
\label{eq:fg_vectors}
f(x) = \begin{pmatrix} x_2 \\ B(x_1 x_4^2 - G\sin x_3) \\ x_4 \\ 0 \end{pmatrix}\!, \quad g(x) = \begin{pmatrix} 0 \\ 0 \\ 0 \\ 1 \end{pmatrix}\!.
\end{equation}
The \emph{admissible region} $\Xadm := \{x\in\R^4 : |x_3| < \pi/2\}$ is the set of states in which the beam has not passed through the vertical. 
It contains every configuration of practical interest.

%======================================================================
\section{Input--Output Linearisation and Singularity Analysis}
\label{sec:singularity}
%======================================================================

We now apply the input--output linearisation procedure to the ball-and-beam and identify the singularity structure.

\subsection{Computation of relative degree}

Starting from $y = h(x) = x_1$, we differentiate until $u$ appears.

\textbf{First derivative:}
$\dot{y} = L_fh(x) = x_2$; \quad $L_gh(x) = 0$.

\textbf{Second derivative:}
$\ddot{y} = L_f^2h(x) = Bx_1 x_4^2 - BG\sin x_3$; \quad $L_gL_fh(x) = 0$.

\textbf{Third derivative:} With $\nabla(L_f^2h) = (Bx_4^2, 0, -BG\cos x_3, 2Bx_1 x_4)$:
\begin{align}
L_f^3h(x) &= Bx_2 x_4^2 - BGx_4\cos x_3 =: b(x), \\
L_gL_f^2h(x) &= 2Bx_1 x_4 =: a(x).
\end{align}

Thus:
\begin{equation}
\label{eq:third_derivative}
y^{(3)} = \underbrace{Bx_2 x_4^2 - BGx_4\cos x_3}_{b(x)} + \underbrace{2Bx_1 x_4}_{a(x)} \cdot u.
\end{equation}
The system has relative degree $\gamma = 3$ whenever $a(x) = 2Bx_1 x_4 \neq 0$, and exhibits a singularity at the third output derivative in the sense of \Cref{def:m_singularity}.

\subsection{Structure of the singularity manifold}

The singularity manifold is
\begin{equation}
\Sing = \{x \in \R^4 : 2Bx_1 x_4 = 0\} = \underbrace{\{x : x_1 = 0\}}_{\Sing_1} \cup \underbrace{\{x : x_4 = 0\}}_{\Sing_2}.
\end{equation}

The control coefficient has the factored form
\begin{equation}
\label{eq:factored_form}
a(x) = 2B \cdot \varphi_1(x) \cdot \varphi_2(x), \quad \varphi_1(x) := x_1, \;\; \varphi_2(x) := x_4,
\end{equation}
with $c(x) = 2B = 10/7 \neq 0$.

\begin{proposition}[Ball-and-beam: regular 2-component singularity manifold]
\label{prop:two_component}
$\Sing$ is a regular $k = 2$ component singularity manifold in the sense of \Cref{def:k_component}. The two components intersect transversally along the two-dimensional plane $\Sing_1\cap\Sing_2 = \{x_1 = x_4 = 0\}$.
\end{proposition}

\begin{proof}
(i) $a(x) = 2B \cdot x_1 \cdot x_4 = c\,\varphi_1\,\varphi_2$ with $c = 2B \neq 0$.
(ii) $\Sing_1$ and $\Sing_2$ are hyperplanes, hence nonempty connected regular hypersurfaces, with $\dd\varphi_1 = (1,0,0,0)$ and $\dd\varphi_2 = (0,0,0,1)$.
(iii) These covectors are linearly independent everywhere, so the family is in general position, and \Cref{lem:intersections} gives $\operatorname{codim}(\Sing_1\cap\Sing_2) = 2$.
\end{proof}

\subsection{Physical interpretation}

\textbf{$\Sing_1$: Ball at the pivot ($x_1 = 0$).} When the ball sits at the pivot, the moment arm for controlling its acceleration via beam rotation vanishes. No matter how fast the beam rotates, the centrifugal effect on the ball is zero at $r = 0$.

\textbf{$\Sing_2$: Zero beam angular velocity ($x_4 = 0$).} When the beam is momentarily stationary, there is no instantaneous coupling between beam acceleration $u = \ddot{\theta}$ and ball acceleration through the third derivative of ball position. The effect of beam acceleration appears only in higher derivatives.

\subsection{Involutivity analysis}

We verify that full state linearisation is impossible, confirming \citet{Hauser1992}. For a fourth-order system, full state linearisation requires the distribution $\Delta = \mathrm{span}\{g, \mathrm{ad}_f g, \mathrm{ad}_f^2 g\}$ to be involutive. The Jacobian of $f$ is
\begin{equation}
\frac{\partial f}{\partial x} = \begin{pmatrix} 0 & 1 & 0 & 0 \\ Bx_4^2 & 0 & -BG\cos x_3 & 2Bx_1 x_4 \\ 0 & 0 & 0 & 1 \\ 0 & 0 & 0 & 0 \end{pmatrix}\!,
\end{equation}
giving
\begin{equation}
\mathrm{ad}_f g = (0, -2Bx_1 x_4, -1, 0)^\top, \qquad \mathrm{ad}_f^2 g = (2Bx_1x_4,\, -2Bx_2x_4 - BG\cos x_3,\, 0,\, 0)^\top,
\end{equation}
and
$[g, \mathrm{ad}_f^2 g] = (2Bx_1, -2Bx_2, 0, 0)^\top$. Since
\begin{equation}
\det\bigl[\,g \;\; \mathrm{ad}_f g \;\; \mathrm{ad}_f^2 g \;\; [g,\mathrm{ad}_f^2 g]\,\bigr] = 2B^2 G\, x_1\cos x_3 ,
\end{equation}
the bracket $[g,\mathrm{ad}_f^2 g]$ does not lie in $\Delta$ wherever $x_1\cos x_3\neq 0$, confirming the failure of involutivity.

\subsection{The obstructed set}

The two singular components are not equally severe. To see this we compute the fourth output derivative of the \emph{true} system along an arbitrary $C^1$ input signal $u(t)$. Differentiating~\eqref{eq:third_derivative} along~\eqref{eq:statespace} gives
\begin{equation}
\label{eq:fourth_derivative}
y^{(4)} = \underbrace{Bx_4^2\bigl[Bx_1x_4^2 + G(1-B)\sin x_3\bigr]}_{\beta(x)} + \bigl(4Bx_2x_4 - BG\cos x_3\bigr)u + 2Bx_1u^2 + 2Bx_1x_4\,\dot u .
\end{equation}
Define the sets
\begin{align}
\Res_1 &:= \Sing_1\cap\{\cos x_3 = 0\} = \{x : x_1 = 0,\ \cos x_3 = 0\}, \\
\Res_2 &:= \Sing_2\cap\{\cos x_3 = 0\} = \{x : x_4 = 0,\ \cos x_3 = 0\}, \\
\Res  &:= \Res_1\cup\Res_2 = \Sing\cap\{\cos x_3 = 0\}, \\
\Obs  &:= \Res_2 \;\cup\; \{x : x_1 = x_2 = 0,\ \cos x_3 = 0\}. \label{eq:obstructed_set}
\end{align}
Each of $\Res_1,\Res_2$ is a countable union of two-dimensional affine planes (one for each root of $\cos x_3 = 0$), hence of codimension two in $\R^4$; $\Obs\subset\Res$ is the union of the planes $\Res_2$ and of the lines $\{x_1=x_2=0,\ x_3 = \pi/2 + j\pi\}$, $j\in\Z$. Physically, $\Res_2$ is the set of states with the beam vertical and momentarily at rest, and the second piece of $\Obs$ has the beam vertical and the ball at rest at the pivot.

\begin{proposition}[Obstruction]
\label{prop:obstruction}
At every point $x\in\Obs$ and for every $C^1$ input signal, the output derivatives of the true system satisfy
\begin{equation}
y^{(3)} = L_f^3h(x), \qquad y^{(4)} = \beta(x) + 2Bx_1u^2 ,
\end{equation}
i.e.\ $y^{(3)}$ is independent of the input and $y^{(4)}$ is an even function of the input, with image $\beta(x) + 2Bx_1[0,\infty)$ if $x_1\neq 0$ and $\{\beta(x)\}$ if $x_1 = 0$. Consequently no smooth static feedback $u=\kappa(x,v)$ is a linearising law of any order $\ell\le 4$ at $x$ (\Cref{def:law_order}). Since every exact or approximate linearising law of the fourth-order system~\eqref{eq:statespace} has order at most $4$, no such law can linearise the true output at a point of $\Obs$.
\end{proposition}

\begin{proof}
On $\Obs$ one has $x_1x_4 = 0$ and $\cos x_3 = 0$, so the coefficient $a = 2Bx_1x_4$ of $u$ in~\eqref{eq:third_derivative} vanishes. The coefficient $4Bx_2x_4 - BG\cos x_3 = 4Bx_2x_4$ of $u$ in~\eqref{eq:fourth_derivative} vanishes because $x_4 = 0$ on $\Res_2$ and $x_2 = 0$ on the second piece; and the coefficient $2Bx_1x_4$ of $\dot u$ vanishes. This proves the two formulas. A linearising law of order $\ell$ requires $y^{(\ell)}$ to take every real value as $v$ ranges over $\R$; for $\ell\in\{1,2\}$ this fails because $L_gh = L_gL_fh \equiv 0$, for $\ell = 3$ because $y^{(3)}$ does not depend on $u$, and for $\ell = 4$ because the image of $u\mapsto\beta(x)+2Bx_1u^2$ is a half-line or a point. The last assertion follows from \Cref{def:chain}, according to which every exact or approximate linearising law of the fourth-order system has order $\ell\le 4$.
\end{proof}

On the larger piece $\Res_2$ of the obstructed set the situation is even more rigid. The ball cannot sense the sign of the torque at all.

\begin{proposition}[Reflection symmetry on $\Res_2$]
\label{prop:symmetry}
Let $x^\circ\in\Res_2$, i.e.\ $x_4^\circ = 0$ and $\cos x_3^\circ = 0$, and let $u(\cdot)$ be any locally integrable input. If $x(t)$ is the solution of~\eqref{eq:statespace} from $x^\circ$ under $u(\cdot)$, then
\begin{equation}
\tilde x(t) := \bigl(x_1(t),\, x_2(t),\, 2x_3^\circ - x_3(t),\, -x_4(t)\bigr)
\end{equation}
is the solution from $x^\circ$ under $-u(\cdot)$. In particular the ball motion $r(t) = x_1(t)$ from any state of $\Res_2$ is the same for the inputs $u(\cdot)$ and $-u(\cdot)$, and every output derivative at $x^\circ$ is an even function of the input.
\end{proposition}

\begin{proof}
Since $\cos x_3^\circ = 0$, we have $\sin 2x_3^\circ = 0$ and $\cos 2x_3^\circ = -1$, hence $\sin(2x_3^\circ - x_3) = \sin 2x_3^\circ\cos x_3 - \cos 2x_3^\circ\sin x_3 = \sin x_3$. Therefore $\dot{\tilde x}_1 = \tilde x_2$, $\dot{\tilde x}_2 = B(\tilde x_1\tilde x_4^2 - G\sin\tilde x_3)$, $\dot{\tilde x}_3 = -x_4 = \tilde x_4$ and $\dot{\tilde x}_4 = -u$, so $\tilde x$ solves~\eqref{eq:statespace} with input $-u$. Its initial condition is $(x_1^\circ, x_2^\circ, x_3^\circ, -x_4^\circ) = x^\circ$ because $x_4^\circ = 0$. Uniqueness of solutions completes the proof.
\end{proof}

\Cref{prop:obstruction,prop:symmetry} show that the loss of control authority on $\Obs$ is a property of the physical system and not of any particular linearisation method: whatever approximating model is chosen, its predicted first-order response to the input at a point of $\Obs$ has no counterpart in the true dynamics.

%======================================================================
\section{Construction of the Control Laws}
\label{sec:controllers}
%======================================================================

Since exact full-state linearisation is impossible, we consider the exact input--output linearising law together with the two canonical approximate laws of \citet{Hauser1992}, one obtained by neglecting a drift term and one by neglecting an input term. All three laws are written in the form $u_j(x,v) = [-b_j(x) + v]/a_j(x)$ with explicit smooth $a_j, b_j$.

\subsection{Control Law~1: exact input--output linearisation}

On $\R^4 \setminus \Sing$, where $a(x) = 2Bx_1 x_4 \neq 0$:
\begin{equation}
\label{eq:control_law_1}
u_1(x, v) = \frac{-Bx_2 x_4^2 + BGx_4\cos x_3 + v}{2Bx_1 x_4}.
\end{equation}
This achieves $y^{(3)} = v$ exactly.

\textbf{Validity domain:} $\Val_1 = \{x : x_1 \neq 0 \text{ and } x_4 \neq 0\}$.

\textbf{Singular set:} $\Sing^{(1)} = \Sing = \Sing_1 \cup \Sing_2$.

Since the relative degree is three for a fourth-order system, the exact law leaves a one-dimensional internal dynamics. Along an exactly tracked output $r(t) = r_d(t)$ the beam angle is governed by
\begin{equation}
\label{eq:internal}
\ddot r_d = B\bigl(r_d\dot\theta^2 - G\sin\theta\bigr),
\end{equation}
a first-order implicit differential equation for $\theta$. Its solutions necessarily reach $\dot\theta = 0$, i.e.\ the singular component $\Sing_2$, whenever the task is periodic (\Cref{prop:crossings}); the behaviour of this internal dynamics at the switching surface is precisely the issue studied by \citet{Tomlin1998}.

\subsection{Control Law~2: approximate linearisation by neglecting the centrifugal term ($f$-approximation)}

The centrifugal term $\psi_2(x) := Bx_1 x_4^2 = \tfrac12 x_4\, a(x)$ in $\ddot{y}$ vanishes on $\Sing$. Following \citet{Hauser1992}, neglecting it and differentiating to fourth order yields the chain
\begin{equation}
\xi_1 = x_1, \quad \xi_2 = x_2, \quad \xi_3 = -BG\sin x_3, \quad \xi_4 = -BGx_4\cos x_3,
\end{equation}
whose only nonzero defect is the drift defect $\psi_2$ (all input defects vanish, since $L_g\xi_j = 0$ for $j\le 3$). Computing $\dot{\xi}_4$ along the true system,
\begin{equation}
\dot{\xi}_4 = BGx_4^2\sin x_3 + (-BG\cos x_3) \cdot u,
\end{equation}
which gives the second control law:
\begin{equation}
\label{eq:control_law_2}
u_2(x, v) = \frac{BGx_4^2\sin x_3 - v}{BG\cos x_3}.
\end{equation}

\textbf{Validity domain:} $\Val_2 = \{x : \cos x_3 \neq 0\} \supseteq \Xadm$.

\textbf{Singular set:} $\Sing^{(2)} = \{x : \cos x_3 = 0\} = \bigcup_{j\in\Z}\{x_3 = \pi/2 + j\pi\}$, a countable union of hyperplanes. It is a regular hypersurface since $\dd(\cos x_3) = -\sin x_3\,\dd x_3 \neq 0$ on it.

\textbf{Consistency:} the chain is consistent on all of $\Sing$, because $\psi_2 = \tfrac12 x_4 a(x)$ vanishes there. The family $\{\Sing_1, \Sing_2, \Sing^{(2)}\}$ is in general position. The differentials $(1,0,0,0)$, $(0,0,0,1)$ and $(0,0,-\sin x_3,0)$ with $\sin x_3 = \pm 1$ are linearly independent everywhere on $\Sing\cup\Sing^{(2)}$.

\subsection{Control Law~3: approximate linearisation by neglecting the input term ($g$-approximation)}

Alternatively, the input term $a(x)u = 2Bx_1x_4\,u$ in~\eqref{eq:third_derivative}, which also vanishes on $\Sing$, may be neglected. The corresponding chain consists of the true output derivatives,
\begin{equation}
\zeta_1 = x_1,\quad \zeta_2 = x_2,\quad \zeta_3 = L_f^2h = B(x_1x_4^2 - G\sin x_3),\quad \zeta_4 = L_f^3h = Bx_2x_4^2 - BGx_4\cos x_3,
\end{equation}
whose only nonzero defect is the input defect $\chi_3 = L_g\zeta_3 = a(x)$. Along the true system,
\begin{equation}
\dot\zeta_4 = L_f^4h(x) + L_gL_f^3h(x)\,u = \beta(x) + B\bigl(2x_2x_4 - G\cos x_3\bigr)\,u,
\end{equation}
with $\beta(x) = Bx_4^2[Bx_1x_4^2 + G(1-B)\sin x_3]$ as in~\eqref{eq:fourth_derivative}, which gives the third control law:
\begin{equation}
\label{eq:control_law_3}
u_3(x, v) = \frac{-\beta(x) + v}{B\bigl(2x_2x_4 - G\cos x_3\bigr)} .
\end{equation}

\textbf{Validity domain:} $\Val_3 = \{x : 2x_2x_4 \neq G\cos x_3\}$. It contains every state of $\Xadm$ with $|2x_2x_4| < G\cos x_3$, i.e.\ every state of practical interest in which the product of ball speed and beam rate is moderate.

\textbf{Singular set:} $\Sing^{(3)} = \{x : \varphi_3(x) = 0\}$ with $\varphi_3 := 2x_2x_4 - G\cos x_3$. It is a regular hypersurface: $\dd\varphi_3 = (0, 2x_4, G\sin x_3, 2x_2)$ vanishes only where $x_2 = x_4 = 0$ and $\sin x_3 = 0$, and at such points $\varphi_3 = \mp G \neq 0$.

\textbf{Consistency:} the chain is consistent on all of $\Sing$, because $\chi_3 = a(x)$ vanishes there. Note that $\Sing^{(3)}\supseteq\Res_2$: on $\Res_2$ both $x_4$ and $\cos x_3$ vanish, so $\varphi_3 = 0$. This is not a deficiency of the construction but a manifestation of \Cref{prop:obstruction}.

\subsection{Summary}

\Cref{tab:control_laws} collects the three laws.

\begin{table}[htbp]
\centering
\caption{Control laws for the ball-and-beam system ($n = 4$, $k = 2$).}
\label{tab:control_laws}
\begin{tabular}{@{}cccll@{}}
\toprule
Law & Order & Coefficient $a_j(x)$ & Validity domain $\Val_j$ & Singular set \\
\midrule
1 & 3 & $2Bx_1 x_4$ & $x_1 \neq 0,\; x_4 \neq 0$ & $\Sing = \Sing_1 \cup \Sing_2$ \\
2 & 4 & $-BG\cos x_3$ & $\cos x_3 \neq 0$ & $\Sing^{(2)} = \{\cos x_3 = 0\}$ \\
3 & 4 & $B(2x_2x_4 - G\cos x_3)$ & $2x_2x_4 \neq G\cos x_3$ & $\Sing^{(3)} = \{2x_2x_4 = G\cos x_3\}$ \\
\bottomrule
\end{tabular}
\end{table}

\begin{theorem}[Coverage of the ball-and-beam]
\label{thm:bb}
Let $\Val_1, \Val_2, \Val_3$ be the validity domains in \Cref{tab:control_laws}. Then:
\begin{enumerate}
\item[(a)] $\Val_1\cup\Val_2 = \R^4\setminus\Res$, where $\Res = \Res_1\cup\Res_2$ is a countable union of two-dimensional planes. Moreover $\Val_2\supseteq\Xadm$: the single law $u_2$ covers the admissible region completely.
\item[(b)] $\Val_1\cup\Val_3 = \R^4\setminus\bigl(\Res_2\cup\{x_1 = 0,\ 2x_2x_4 = G\cos x_3\}\bigr)$ and $\Val_2\cup\Val_3 = \R^4\setminus\{\cos x_3 = 0,\ x_2x_4 = 0\}$; both uncovered sets are finite unions of embedded submanifolds of codimension two. Hence any two of the three laws cover $\R^4$ in codimension two.
\item[(c)] $\Val_1\cup\Val_2\cup\Val_3 = \R^4\setminus\Obs$. By \Cref{prop:obstruction}, $\R^4\setminus\Obs$ is the largest set on which linearising laws of order at most four can exist; the family $\{u_1,u_2,u_3\}$ is therefore maximal, and complete coverage of $\R^4$ is impossible within this class.
\item[(d)] Each single law leaves a nonempty regular hypersurface uncovered ($\Sing$, $\Sing^{(2)}$ and $\Sing^{(3)}$ respectively).
\end{enumerate}
\end{theorem}

\begin{proof}
(a) $\R^4\setminus(\Val_1\cup\Val_2) = \Sing\cap\Sing^{(2)} = (\Sing_1\cup\Sing_2)\cap\{\cos x_3 = 0\} = \Res_1\cup\Res_2$. Each $\Res_i$ is the union over $j\in\Z$ of the planes $\{x_1 = 0, x_3 = \pi/2 + j\pi\}$, respectively $\{x_4 = 0, x_3 = \pi/2 + j\pi\}$, which have dimension two. Since $\cos x_3 > 0$ on $\Xadm$, $\Xadm\subseteq\Val_2$.

(b) $\R^4\setminus(\Val_1\cup\Val_3) = \Sing\cap\Sing^{(3)}$. On $\Sing_2$, $x_4 = 0$ forces $\varphi_3 = -G\cos x_3 = 0$, so $\Sing_2\cap\Sing^{(3)} = \Res_2$. On $\Sing_1$, $\Sing_1\cap\Sing^{(3)} = \{x_1 = 0,\ 2x_2x_4 = G\cos x_3\}$. The map $(x_2,x_3,x_4)\mapsto 2x_2x_4 - G\cos x_3$ has nonvanishing differential on its zero set (shown above), so this is a regular surface in $\Sing_1\cong\R^3$, hence an embedded submanifold of codimension two in $\R^4$. Likewise $\R^4\setminus(\Val_2\cup\Val_3) = \{\cos x_3 = 0\}\cap\{2x_2x_4 = G\cos x_3\} = \{\cos x_3 = 0,\ x_2x_4 = 0\}$, the union of the planes $\{x_2 = 0, x_3 = \pi/2 + j\pi\}$ and $\{x_4 = 0, x_3 = \pi/2 + j\pi\}$.

(c) $\R^4\setminus(\Val_1\cup\Val_2\cup\Val_3) = \Res\cap\Sing^{(3)}$. On $\Res$ we have $\cos x_3 = 0$, so $\Sing^{(3)}\cap\Res = \{x\in\Res : x_2x_4 = 0\}$. 
Every point of $\Res_2$ has $x_4 = 0$, so $\Res_2\subseteq\Sing^{(3)}$. A point of $\Res_1$ with $x_2x_4 = 0$ has either $x_4 = 0$, in which case it lies in $\Res_2$, or $x_2 = 0$, in which case it lies in $\{x_1 = x_2 = 0, \cos x_3 = 0\}$. Hence $\Res\cap\Sing^{(3)} = \Res_2\cup\{x_1 = x_2 = 0,\ \cos x_3 = 0\} = \Obs$. The remaining assertions follow from \Cref{prop:obstruction}.

(d) $\Sing$, $\Sing^{(2)}$ and $\Sing^{(3)}$ are nonempty regular hypersurfaces by \Cref{prop:two_component} and the computations above; e.g.\ $(0,0,\pi/2,1)\in\Sing^{(3)}$.
\end{proof}

%======================================================================
\section{Coverage Theory}
\label{sec:theorem}
%======================================================================

We now place the ball-and-beam results in a general framework and prove the codimension theorem that governs the number of control laws.

\subsection{Families of control laws and coverage}

\begin{definition}[Family of linearising laws; uncovered set]
\label{def:family}
A \emph{family of linearising laws} for the system~\eqref{eq:system} is a finite collection $\mathfrak{U} = \{u_1,\ldots,u_p\}$ of exact or approximate linearising laws (\Cref{def:m_singularity,def:chain}) with validity domains $\Val_1,\ldots,\Val_p$ and singular sets $\Sigma_j := \mathcal{X}\setminus\Val_j$. Its \emph{uncovered set} is
\begin{equation}
\Unc(\mathfrak{U}) := \mathcal{X}\setminus\bigcup_{j=1}^p\Val_j = \bigcap_{j=1}^p\Sigma_j .
\end{equation}
The family \emph{covers a region $\mathcal{Y}\subseteq\mathcal{X}$ completely} if $\mathcal{Y}\subseteq\bigcup_j\Val_j$, and it \emph{covers $\mathcal{X}$ in codimension $c$} if $\Unc(\mathfrak{U})$ is empty or a finite union of embedded submanifolds of codimension at least $c$.
\end{definition}

\begin{remark}
Complete coverage is the natural requirement when every state must admit a well-defined law. Coverage in codimension two is the natural requirement for trajectory tracking: by \Cref{lem:thin}, a set of codimension two is reached from a null set of initial states only, whereas a hypersurface (codimension one) is crossed by open sets of trajectories.
\end{remark}

\begin{assumption}[Regular family in general position]
\label{ass:general_position}
Each singular set is a finite union of connected regular hypersurfaces, $\Sigma_j = \bigcup_{i=1}^{k_j}\Sigma_{j,i}$, and the collection $\{\Sigma_{j,i} : 1\le j\le p,\ 1\le i\le k_j\}$ is in general position (\Cref{def:general_position}). For the exact law, $\Sigma_1 = \Sing$ and $k_1 = k$ is the number of components of the singularity manifold.
\end{assumption}

\subsection{The codimension theorem}

\begin{theorem}[Codimension of the uncovered set]
\label{thm:codim}
Let $\mathfrak{U} = \{u_1,\ldots,u_p\}$ satisfy \Cref{ass:general_position}. Then
\begin{equation}
\label{eq:uncovered_decomposition}
\Unc(\mathfrak{U}) = \bigcup_{(i_1,\ldots,i_p)}\ \bigcap_{j=1}^p\Sigma_{j,i_j},
\end{equation}
where the union runs over the $\prod_j k_j$ choices of one component from each singular set, and each set $\bigcap_j\Sigma_{j,i_j}$ is either empty or a closed embedded submanifold of $\mathcal{X}$ of codimension exactly $p$. Consequently:
\begin{enumerate}
\item[(i)] $\mathfrak{U}$ covers $\mathcal{X}$ in codimension $p$;
\item[(ii)] if $p\ge n+1$, then $\Unc(\mathfrak{U}) = \emptyset$ and $\mathfrak{U}$ covers $\mathcal{X}$ completely;
\item[(iii)] if $p\le n$, then $\mathfrak{U}$ covers $\mathcal{X}$ completely if and only if each of the finitely many intersections in~\eqref{eq:uncovered_decomposition} is empty.
\end{enumerate}
\end{theorem}

\begin{proof}
Since $\Sigma_j = \bigcup_i\Sigma_{j,i}$, distributivity of intersection over union gives~\eqref{eq:uncovered_decomposition}. Each intersection $\bigcap_j\Sigma_{j,i_j}$ involves exactly $p$ members of a family in general position, so \Cref{lem:intersections} shows that it is empty or a closed embedded submanifold of codimension $p$, and that it is empty when $p>n$. Statements (i)--(iii) follow.
\end{proof}

\begin{remark}
\label{rem:k_role}
For the most common family, the exact law together with $p-1$ approximate laws whose singular sets are single regular hypersurfaces $\Sing^{(2)},\ldots,\Sing^{(p)}$, the uncovered set has at most $k$ pieces, one inside each component of the singularity manifold:
\begin{equation}
\Unc(\mathfrak{U}) = \bigcup_{i=1}^k\Bigl(\Sing_i\cap\Sing^{(2)}\cap\cdots\cap\Sing^{(p)}\Bigr).
\end{equation}
The number $k$ of components thus determines the number of residual pieces, whereas the codimension of the residual is determined by the number $p$ of laws. In particular, one approximate law whose singular set is transverse to every component of $\Sing$ already reduces the uncovered set to codimension two, regardless of $k$; for the ball-and-beam ($k=2$) this is the content of \Cref{thm:bb}(a).
\end{remark}

\begin{remark}
\label{rem:wellposed}
Coverage in \Cref{def:family} refers to the domain on which a law is well defined. \Cref{thm:codim}(ii) therefore asserts that $n+1$ laws in general position leave no state without a well-defined law. For the ball-and-beam, \Cref{prop:obstruction} shows that no law can \emph{linearise the true output} on $\Obs$. Consequently any law that is well defined on $\Obs$ is derived from a model whose input coupling there is fictitious, and any family covering $\R^4$ completely must contain such a law. The set $\R^4\setminus\Obs$ is the largest set that can be covered by laws of order at most four, and it is covered by the three laws of \Cref{tab:control_laws}.
\end{remark}

\begin{corollary}[Two laws suffice for generic trajectories]
\label{cor:two_laws}
Let $\mathfrak{U} = \{u_1,u_2\}$ satisfy \Cref{ass:general_position}. Then $\Unc(\mathfrak{U})$ is a finite union of embedded submanifolds of codimension two, and for every smooth vector field $F$ defined on an open subset $\mathcal{W}\subseteq\mathcal{X}$ (for instance the closed loop under either law, on the interior of its validity domain), the set of initial states in $\mathcal{W}$ whose trajectory under $F$ reaches $\Unc(\mathfrak{U})$ has Lebesgue measure zero.
\end{corollary}

\begin{proof}
Combine \Cref{thm:codim}(i) with \Cref{lem:thin}.
\end{proof}

\subsection{Necessity of a second law for the ball-and-beam}

The exact law is undefined on $\Sing$, and its gain $1/a(x)$ is unbounded near $\Sing$. The following result shows that $\Sing_2$ cannot be avoided by any trajectory that performs a periodic task.

\begin{proposition}[Unavoidable crossings of $\Sing_2$]
\label{prop:crossings}
Let $x(t)$, $t\ge 0$, be a solution of~\eqref{eq:statespace} with a bounded measurable input $u$ and bounded beam angle $x_3$. Let $r_d$ be a non-constant periodic function of class $C^2$ and suppose that $r(t) - r_d(t)\to 0$ and $\ddot r(t) - \ddot r_d(t)\to 0$ as $t\to\infty$ (in particular, this holds if $r\equiv r_d$). Then $\dot\theta(t) = x_4(t)$ vanishes for arbitrarily large $t$: the trajectory meets $\Sing_2$ infinitely often.
\end{proposition}

\begin{proof}
Suppose, for contradiction, that $x_4(t)\neq 0$ for all $t\ge T$. Then $x_3$ is monotone on $[T,\infty)$; being bounded, it converges to some $\theta_\infty$. Since $\dot x_4 = u$ is bounded, $x_4$ is uniformly continuous, and Barbalat's lemma yields $x_4(t)\to 0$. As $r_d$ is bounded and $r - r_d\to 0$, $x_1$ is bounded, so $\ddot r = B(x_1x_4^2 - G\sin x_3)\to -BG\sin\theta_\infty$. Hence $\ddot r_d(t)\to -BG\sin\theta_\infty$. A continuous periodic function with a limit at infinity is constant, so $\ddot r_d\equiv c$; then $\dot r_d(t) = ct + \dot r_d(0)$ is periodic only if $c = 0$, and then $r_d$ is affine and periodic, hence constant, a contradiction.
\end{proof}

\begin{corollary}[Minimum number of laws for the ball-and-beam]
\label{cor:minimum}
For the ball-and-beam system~\eqref{eq:statespace}:
\begin{enumerate}
\item[(i)] No family consisting of the exact law alone can execute a periodic tracking task with bounded input and bounded beam angle: by \Cref{prop:crossings} the trajectory must pass through $\Sing_2$, where $u_1$ is undefined. At least one approximate law is therefore necessary.
\item[(ii)] On the admissible region $\Xadm$, one law suffices: $u_2$ covers $\Xadm$ completely (\Cref{thm:bb}(a)).
\item[(iii)] On $\R^4$, two laws are necessary and sufficient for coverage in codimension two: each single law leaves a hypersurface uncovered (\Cref{thm:bb}(d)), and any two of $u_1,u_2,u_3$ cover $\R^4$ in codimension two (\Cref{thm:bb}(a),(b)).
\item[(iv)] On $\R^4$, three laws cover the maximal coverable set $\R^4\setminus\Obs$ (\Cref{thm:bb}(c)); complete coverage of $\R^4$ by linearising laws of order at most four is impossible (\Cref{prop:obstruction}).
\end{enumerate}
\end{corollary}

%======================================================================
\section{Closed-Loop Architecture}
\label{sec:architecture}
%======================================================================

The implementation of a family of control laws requires a supervisory mechanism selecting, at every state, a law whose validity domain contains that state. For the ball-and-beam family of \Cref{tab:control_laws} a natural supervisor is
\begin{equation}
\label{eq:supervisor}
\sigma(x) = \begin{cases}
1 & \text{if } |a(x)| = |2Bx_1x_4| \ge \delta_1, \\
2 & \text{if } |a(x)| < \delta_1 \text{ and } |\cos x_3| \ge \delta_2, \\
3 & \text{if } |a(x)| < \delta_1,\ |\cos x_3| < \delta_2 \text{ and } |2x_2x_4 - G\cos x_3| \ge \delta_3,
\end{cases}
\end{equation}
with thresholds $\delta_1,\delta_2,\delta_3 > 0$. By \Cref{thm:bb}, as the thresholds tend to zero the domain of $\sigma$ exhausts $\R^4\setminus\Obs$. The states left without a law for fixed thresholds form a thin neighbourhood of $\Obs$, on which no linearising law exists in any case. On the admissible region one has $|\cos x_3| \ge \cos\theta_{\max}$ whenever $|x_3|\le\theta_{\max}<\pi/2$, so for $\delta_2 \le \cos\theta_{\max}$ only the values $\sigma\in\{1,2\}$ occur, and the choice $\delta_1 = \infty$ (i.e.\ $\sigma\equiv 2$) is admissible: the single approximate law $u_2$ is well defined along every admissible trajectory.

The inclusion of the exact law $u_1$ in the family is dictated by the coverage theory of \Cref{sec:theorem} and by its exactness away from $\Sing$. Whether switching into $u_1$ is beneficial for a given task depends on the internal dynamics~\eqref{eq:internal} at the switching surface, as analysed by \citet{Tomlin1998}. Since the exact law is used only where $|a(x)|\ge\delta_1$, its gain is bounded by $1/\delta_1$, and the switching surface $\{|a(x)| = \delta_1\}$ is a regular hypersurface transverse to $\Sing$ for every $\delta_1>0$. Any approximate law selected near $\Sing$ is consistent there (its defects vanish on $\Sing$), so the approximation error of \citet{Hauser1992} is small exactly where the approximate law is used. Stability of the switched closed loop can be assessed with the tools of \citet{Liberzon2003}. In the simulations reported below only the approximate laws are active along the tracked trajectories, so that no switching occurs.

\begin{figure}[htbp]
\centering
\begin{tikzpicture}[scale=0.85]
    % ---- left panel: (x1,x4) projection
    \begin{scope}
    \draw[->,thick] (-3.5,0) -- (3.5,0) node[right] {$x_1$};
    \draw[->,thick] (0,-3.5) -- (0,3.5) node[above] {$x_4$};
    \fill[green!10] (0.15,0.15) rectangle (3,3);
    \fill[green!10] (-3,-3) rectangle (-0.15,-0.15);
    \fill[green!10] (0.15,-3) rectangle (3,-0.15);
    \fill[green!10] (-3,0.15) rectangle (-0.15,3);
    \draw[very thick, red] (0,-3.2) -- (0,3.2);
    \node[red, fill=white] at (0.7,2.9) {$\Sing_1$};
    \draw[very thick, blue] (-3.2,0) -- (3.2,0);
    \node[blue, fill=white] at (2.9,-0.4) {$\Sing_2$};
    \node at (1.6,1.6) {\small Law~1};
    \node at (-1.6,1.6) {\small Law~1};
    \node at (1.6,-1.6) {\small Law~1};
    \node at (-1.6,-1.6) {\small Law~1};
    \fill[black] (0,0) circle (3pt);
    \node[fill=white] at (0.9,-0.75) {\small$\Sing_1 \!\cap\! \Sing_2$};
    \node[green!50!black, fill=white, align=center] at (-1.9,2.6) {\scriptsize Laws~2,3 on $\Sing$\\[-2pt]\scriptsize (where $\cos x_3\neq 0$)};
    \node at (0,-4) {\small (a) projection onto $(x_1,x_4)$};
    \end{scope}
    % ---- right panel: (x3,x4) projection
    \begin{scope}[xshift=8.5cm]
    \draw[->,thick] (-3.5,0) -- (3.5,0) node[right] {$x_3$};
    \draw[->,thick] (0,-3.5) -- (0,3.5) node[above] {$x_4$};
    \fill[yellow!15] (-2,-3) rectangle (2,3);
    \node[orange!70!black] at (0,2.7) {\small $\Xadm$: $|x_3|<\pi/2$};
    \draw[very thick, blue] (-3.2,0) -- (3.2,0);
    \node[blue, fill=white] at (2.9,-0.4) {$\Sing_2$};
    \draw[very thick, violet] (-2,-3.2) -- (-2,3.2);
    \draw[very thick, violet] (2,-3.2) -- (2,3.2);
    \node[violet, fill=white] at (2.75,2.6) {$\Sing^{(2)}$};
    \node[violet, fill=white] at (-2.75,2.6) {$\Sing^{(2)}$};
    \fill[black] (2,0) circle (3pt); \fill[black] (-2,0) circle (3pt);
    \node[fill=white] at (2.55,-0.9) {\small$\Res_2\subset\Obs$};
    \node[fill=white] at (-2.55,-0.9) {\small$\Res_2\subset\Obs$};
    \draw[thin,gray] (-2,-3.4) node[below] {\scriptsize $-\pi/2$} -- (-2,-3.2);
    \draw[thin,gray] (2,-3.4) node[below] {\scriptsize $\pi/2$} -- (2,-3.2);
    \node at (0,-4) {\small (b) projection onto $(x_3,x_4)$};
    \end{scope}
\end{tikzpicture}
\caption{Singularity structure of the ball-and-beam. (a) The components $\Sing_1 = \{x_1 = 0\}$ (red) and $\Sing_2 = \{x_4 = 0\}$ (blue) meet transversally. Law~1 is valid in the four open quadrants, Laws~2 and~3 are valid on $\Sing$ wherever $\cos x_3\neq 0$. (b) The singular set $\Sing^{(2)} = \{\cos x_3 = 0\}$ of Law~2 (violet) is transverse to $\Sing_2$. Their intersection $\Res_2$ (black dots, a two-dimensional plane in $\R^4$) belongs to the obstructed set $\Obs$, on which no linearising law exists (\Cref{prop:obstruction}). The admissible region $\Xadm$ (shaded) lies entirely inside the validity domain of Law~2.}
\label{fig:control_regions}
\end{figure}

%======================================================================
\section{Numerical Simulations}
\label{sec:simulations}
%======================================================================

Simulations use the parameters of \citet{Hauser1992}: $M = 0.05$~kg, $R = 0.01$~m, $J = 0.02$~kg$\cdot$m$^2$, $J_b = 2 \times 10^{-6}$~kg$\cdot$m$^2$, $G = 9.81$~m/s$^2$, giving $B = 5/7 \approx 0.7143$. The system~\eqref{eq:statespace} is integrated with a fixed-step fourth-order Runge--Kutta scheme (step $10^{-3}$~s) for the closed-loop runs and with an adaptive Runge--Kutta scheme (relative tolerance $10^{-10}$) for the runs that approach $\Sing$. The outer loops use pole placement: for Law~1, $v = y_d^{(3)} - \sum_{i=0}^{2} k_i\,(y^{(i)} - y_d^{(i)})$ with all poles at $s = -4$; for Laws~2 and~3, $v = y_d^{(4)} - \sum_{i=0}^{3} k_i\,(\phi_{i+1} - y_d^{(i)})$ with all poles at $s = -3$, where $\phi$ denotes the respective chain ($\xi$ or $\zeta$). No input saturation is applied.

\subsection{The exact law alone}

\Cref{fig:exact_only} illustrates \Cref{prop:crossings} and \Cref{cor:minimum}(i). The reference $y_d(t) = 0.25 + 0.15\cos(2\pi t/3)$~m keeps the ball on one side of the pivot, so that $\Sing_1$ is never approached and only the mechanism of $\Sing_2$ is at work. The initial state is taken on the reference ($r(0) = 0.4$, $\dot r(0) = 0$, $\ddot r(0) = \ddot y_d(0)$) with $\dot\theta(0) = \pm 0.3$~rad/s, so that the exact law tracks $y_d$ with zero error from the outset (the observed error is below $10^{-14}$~m). The beam angle then follows the internal dynamics~\eqref{eq:internal}, whose two branches behave very differently: for $\dot\theta(0) = -0.3$ the beam settles gently, whereas for $\dot\theta(0) = +0.3$ it swings through $185^\circ$. Both branches reach $|a(x)| < 10^{-5}$ at $t = 1.50$~s, precisely the instant at which $y_d^{(3)}$ vanishes: on $\Sing_2$ the true third derivative $y^{(3)} = L_f^3h$ is zero irrespective of the input, so exact tracking can only cross $\Sing_2$ where $y_d^{(3)} = 0$. Beyond this instant the exact law does not exist. The law is moreover extremely sensitive away from the reference: with the initial ball position displaced by $5$~mm, $|u_1|$ exceeds $10^4$~rad/s$^2$ within $18$~ms, and with $\dot\theta(0)$ displaced by $-0.05$~rad/s it exceeds $10^3$~rad/s$^2$ within $27$~ms, in both cases because the closed loop is driven into $\Sing_2$ almost immediately.

\begin{figure*}[htbp]
\centering
\includegraphics[width=\textwidth]{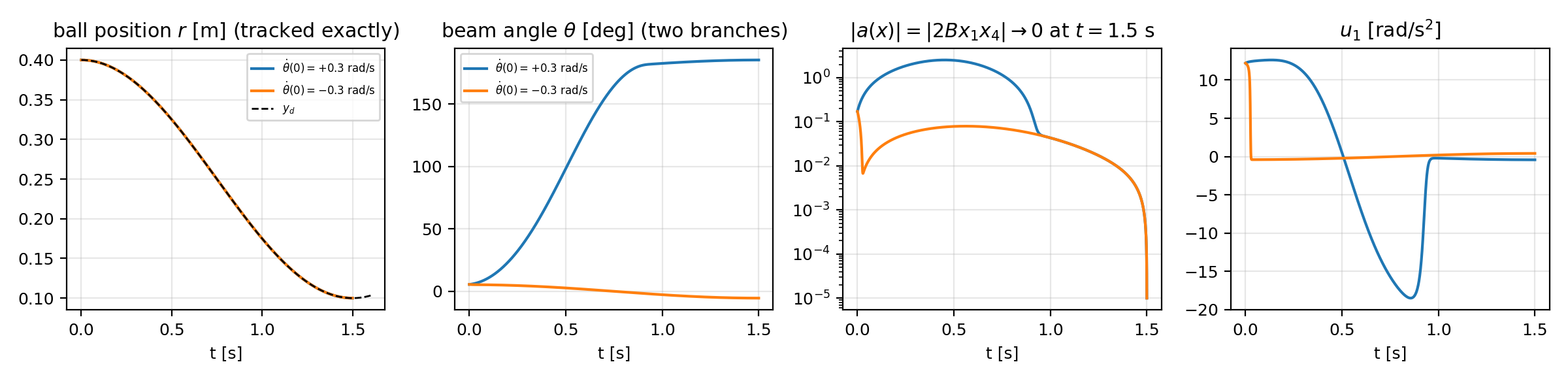}
\caption{Exact law $u_1$ alone, reference $y_d = 0.25 + 0.15\cos(2\pi t/3)$, initial state on the reference with $\dot\theta(0) = \pm0.3$~rad/s. Left to right: ball position (tracked exactly). Beam angle, showing the two branches of the internal dynamics~\eqref{eq:internal}. Control coefficient $|a(x)|$, which vanishes at $t = 1.5$~s, control input. The law ceases to exist when the trajectory reaches $\Sing_2$.}
\label{fig:exact_only}
\end{figure*}

\subsection{A single approximate law on the admissible region}

\Cref{fig:single_law} shows Law~2 alone tracking $y_d(t) = 0.4\cos(2\pi t / 3)$~m over $t \in [0, 30]$~s from the initial state $x(0) = 0$, which lies on $\Sing_1\cap\Sing_2$ (ball at rest at the pivot, beam horizontal and at rest). The reference passes through the pivot twice per period and the beam reverses its rotation twice per period, so the trajectory crosses $\Sing_1$ and $\Sing_2$ twice per period each, as the zero crossings of $a(x)$ in the figure show. Throughout the run $|\theta|\le 14.7^\circ$ and $|\cos x_3|\ge 0.967$, so the trajectory remains far from $\Sing^{(2)}$ and Law~2 is valid at every instant, in agreement with \Cref{thm:bb}(a). The steady-state tracking error (evaluated for $t\ge 10$~s) has maximum $7.1$~mm and root-mean-square value $4.6$~mm, and the input satisfies $|u|\le 7.8$~rad/s$^2$. Law~3 alone produces the same qualitative behaviour with maximum error $6.2$~mm (rms $4.0$~mm); its coefficient satisfies $|2x_2x_4|\le 1.0 \ll G\cos x_3$ along the trajectory. The residual error is the approximation error of \citet{Hauser1992}: it stems from the neglected defect ($\psi_2 = Bx_1x_4^2$ for Law~2, $\chi_3 u = 2Bx_1x_4u$ for Law~3), which is largest when the ball is far from the pivot and the beam rotates fast. \Cref{tab:period} confirms that the error decreases rapidly as the reference is slowed down, while the two singular components are still crossed twice per period.

\begin{figure*}[htbp]
\centering
\includegraphics[width=\textwidth]{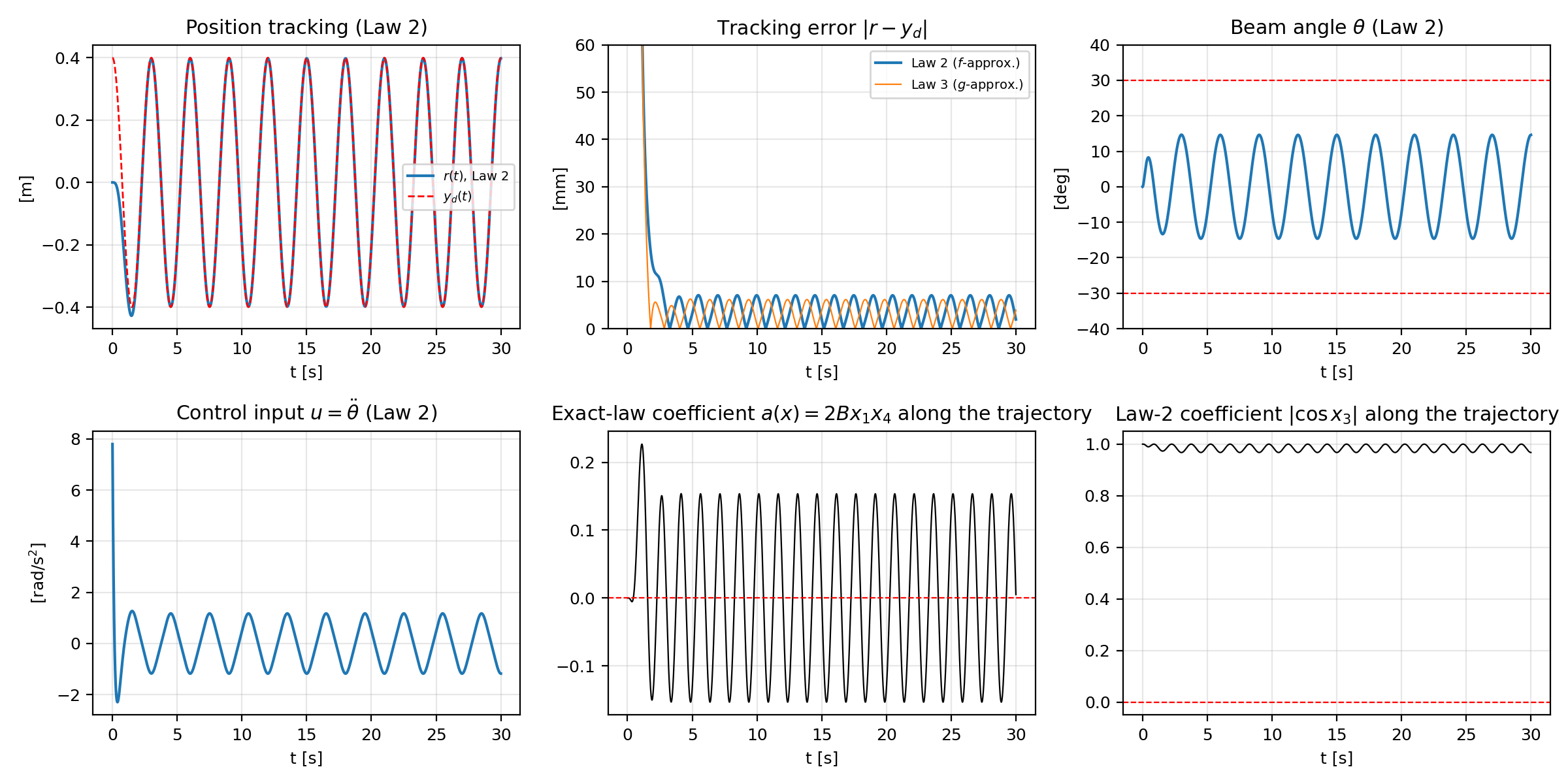}
\caption{Law~2 alone, reference $y_d = 0.4\cos(2\pi t/3)$, initial state $x(0) = 0\in\Sing_1\cap\Sing_2$. Top row: position tracking. Tracking error for Law~2 and, for comparison, for Law~3 alone. Beam angle with the $\pm30^\circ$ bounds. Bottom row: control input; exact-law coefficient $a(x) = 2Bx_1x_4$ along the trajectory, whose zero crossings mark the passages through $\Sing_1$ and $\Sing_2$. Law-2 coefficient $|\cos x_3|$, which stays above $0.96$.}
\label{fig:single_law}
\end{figure*}

\begin{table}[htbp]
\centering
\caption{Steady-state tracking error ($t\ge 10$~s) for $y_d = 0.4\cos(2\pi t/T)$ under a single approximate law, from $x(0) = 0$.}
\label{tab:period}
\small
\begin{tabular}{@{}cccccc@{}}
\toprule
$T$ [s] & Law~2 error max / rms [mm] & Law~3 error max / rms [mm] & $\max|\theta|$ [deg] & $\min|\cos x_3|$ & $\Sing_2$ crossings per period \\
\midrule
2  & 64.1 / 30.0 & 58.5 / 42.1 & 35.1 & 0.818 & 2 \\
3  & 7.05 / 4.58 & 6.16 / 4.00 & 14.7 & 0.967 & 2 \\
5  & 0.49 / 0.33 & 0.28 / 0.18 & 5.2  & 0.996 & 2 \\
10 & 0.01 / 0.01 & $<0.01$ / $<0.01$ & 2.4 & 0.999 & 2 \\
\bottomrule
\end{tabular}
\end{table}

\subsection{Singularity manifold analysis}

\Cref{fig:singularity_traces} presents the distances to the singular components along trajectories of Law~2 starting near $\Sing_1$ ($x(0) = (0, 0, 0, 0.3)$), near $\Sing_2$ ($x(0) = (0.3, 0, 0.2, 0)$) and far from $\Sing$ ($x(0) = (-0.3, 0.2, -0.3, 0.5)$). All three converge to the same periodic orbit, along which $|x_1|$ and $|x_4|$ fall to $10^{-3}$--$10^{-4}$ at each crossing and $|a(x)|$ falls below $10^{-4}$, so that the exact law would be unusable there, while $|\cos x_3|\ge 0.955$ keeps Law~2 valid. The residual set $\Res$ of \Cref{thm:bb}(a) is never approached, as the codimension count of \Cref{cor:two_laws} predicts.

\begin{figure*}[htbp]
\centering
\includegraphics[width=\textwidth]{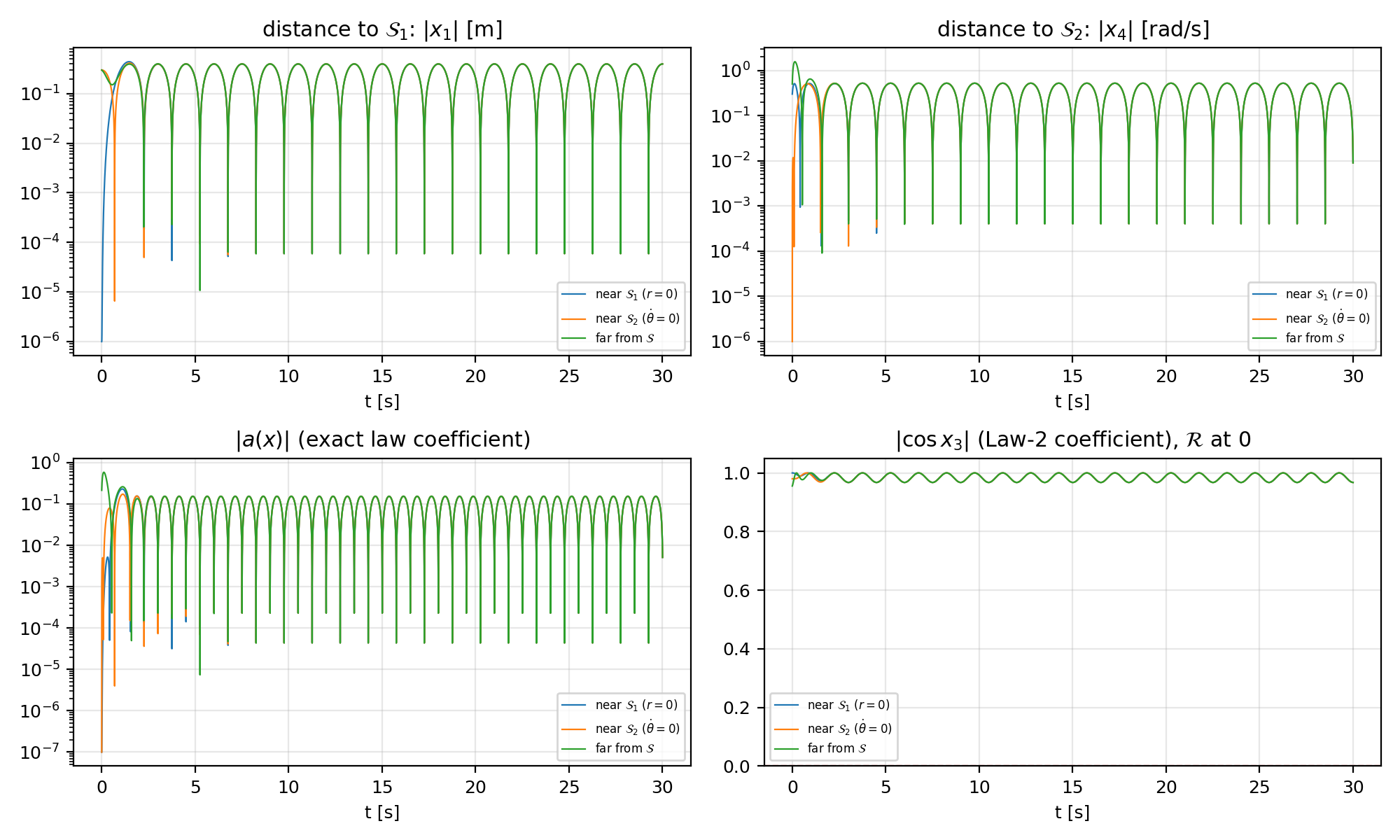}
\caption{Singularity manifold analysis under Law~2 for three initial conditions. Top row: distance to $\Sing_1$ ($|x_1|$) and to $\Sing_2$ ($|x_4|$). Bottom row: exact-law coefficient $|a(x)|$, which vanishes at every crossing of $\Sing$, and Law-2 coefficient $|\cos x_3|$, which stays bounded away from zero: the codimension-two residual $\Res = \Sing\cap\Sing^{(2)}$ is never met.}
\label{fig:singularity_traces}
\end{figure*}

\subsection{The obstructed set}

\Cref{fig:obstruction} illustrates \Cref{prop:symmetry}. From the state $x^\circ = (0.1, 0, \pi/2, 0)\in\Res_2$ (beam vertical and at rest, ball $10$~cm above the pivot) the open-loop system is driven with the constant inputs $u\equiv +10$, $u\equiv -10$ and $u\equiv 0$~rad/s$^2$. The beam angles for $u = \pm 10$ diverge in opposite directions, but the ball motions coincide to machine precision (maximum difference $2\times10^{-16}$~m over $0.6$~s): the ball cannot sense the sign of the torque. The difference between the ball motions for $u = 10$ and $u = 0$ grows as $2Bx_1u^2t^4/4!$, exactly the even term $2Bx_1u^2$ of~\eqref{eq:fourth_derivative}. No feedback law can therefore assign $y^{(4)}$ freely at $x^\circ$, in agreement with \Cref{prop:obstruction}.

\begin{figure*}[htbp]
\centering
\includegraphics[width=\textwidth]{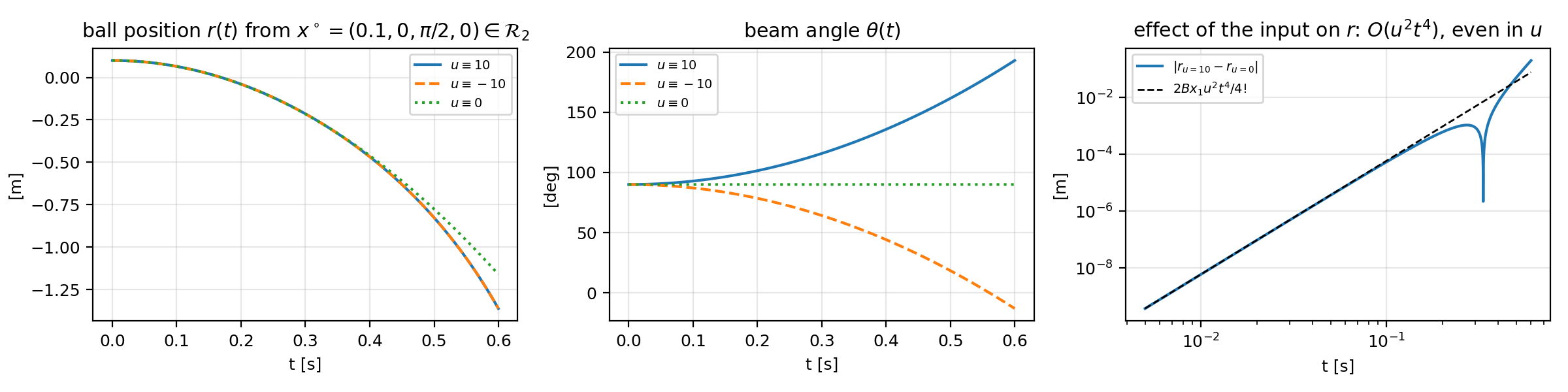}
\caption{Input-reversal symmetry on the obstructed set. Open-loop responses from $x^\circ = (0.1, 0, \pi/2, 0)\in\Res_2$ under $u\equiv 10$, $u\equiv -10$ and $u\equiv 0$: the ball motions for $\pm 10$ are identical (left), while the beam angles are mirror images (centre). The effect of the input on the ball is even in $u$ and of order $t^4$ (right).}
\label{fig:obstruction}
\end{figure*}

The simulations confirm: (1)~the exact law alone loses validity in finite time on every periodic task, as \Cref{prop:crossings} predicts. (2)~a single approximate law covers the admissible region and tracks with an error that vanishes as the reference is slowed down; (3)~the trajectories cross the codimension-one set $\Sing$ twice per period but never meet the codimension-two residual $\Res$; (4)~on the obstructed set $\Obs$ the ball motion is invariant under input reversal.

%======================================================================
\section{Discussion}
\label{sec:discussion}
%======================================================================

\subsection{Geometric interpretation}

The results have a transparent geometric reading. The factored coefficient $a(x) = c(x)\prod_{i=1}^k\varphi_i(x)$ creates $k$ hypersurface ``obstacles'' to exact linearisation. Each additional law whose singular set is transverse to the obstacles cuts the uncovered set down by one dimension: with $p$ laws in general position the residual has codimension $p$ and consists of at most $k$ pieces, one inside each obstacle (\Cref{thm:codim}, \Cref{rem:k_role}). The number of components $k$ therefore governs the \emph{shape} of the residual, not its \emph{codimension}; the latter is governed by the number of laws. Since trajectories cross hypersurfaces but avoid sets of codimension two, two laws already suffice for tracking (\Cref{cor:two_laws}), and the exact law alone never does when the task is periodic (\Cref{prop:crossings}). For the ball-and-beam the single centrifugal defect $\psi_2 = \tfrac12 x_4a(x)$ vanishes on both components at once, so one approximate law is transverse to the whole of $\Sing$, and on the admissible region it is even valid everywhere.

The obstructed set $\Obs$ shows that complete coverage of the whole state space is not merely hard but impossible for this system: on $\Obs$ the ball motion is an even function of the input (\Cref{prop:symmetry}), so no static feedback can linearise the output there at any order up to four (\Cref{prop:obstruction}). The three laws of \Cref{tab:control_laws} cover exactly the complement of $\Obs$ and are therefore a maximal family.

\subsection{Relationship to existing results}

\textbf{Approximate linearisation}~\citep{Hauser1992}: provides the two approximate laws used here and the error bounds that make their use near $\Sing$ meaningful; the present work quantifies where the laws are defined and where no law can be defined.

\textbf{Switching through singularities}~\citep{Tomlin1998}: proposes the switch between the exact law away from $\Sing$ and the approximate law near it and shows that its applicability hinges on the internal dynamics at the switching surface; \Cref{prop:crossings} and \Cref{fig:exact_only} show why the switch is unavoidable, and \Cref{thm:codim} shows what the switched family leaves uncovered.

\textbf{Patchy feedback theory}~\citep{Ancona1999}: guarantees finite controller families for stabilisation but provides no bounds tied to the singularity structure; \Cref{thm:codim} provides such a bound in terms of the codimension of the residual.

\textbf{Switched systems theory}~\citep{Liberzon2003}: provides stability tools for switched systems once the controllers and switching regions are given; it does not address which regions can be covered.

\subsection{Scope and limitations}

Coverage in the sense of \Cref{def:family} concerns the existence of a well-defined law at each state. The accuracy of an approximate law is a separate matter, governed by the size of its defects along the trajectory~\citep{Hauser1992}. Consistency of the chains on $\Sing$ guarantees that the defects are small exactly where the approximate laws are used, and \Cref{tab:period} quantifies the resulting error for the ball-and-beam.

\Cref{thm:codim} rests on \Cref{ass:general_position}, which must be verified for each family. For the ball-and-beam the family $\{u_1,u_2\}$ satisfies it everywhere, whereas $\{u_1,u_2,u_3\}$ satisfies it only away from $\Res_2$, precisely because $\Sing^{(3)}\supseteq\Res_2$. The obstruction of \Cref{prop:obstruction} is specific to the ball-and-beam. Whether analogous obstructed sets exist for other systems with singular relative degree must be established case by case.

The exact law retains a one-dimensional internal dynamics whose behaviour determines whether switching into it is advantageous~\citep{Tomlin1998}. The coverage results are independent of this question, but a complete switched design must address it together with the stability of the switched loop~\citep{Liberzon2003}. Additional limitations include the restriction to SISO systems, smooth nonlinearities, control-affine structure and static feedback. Extensions to MIMO systems with decoupling matrix singularities, to dynamic feedback and to input-constrained systems remain open.

%======================================================================
\section{Conclusions}
\label{sec:conclusion}
%======================================================================

This paper characterised the number of feedback-linearising control laws required for nonlinear systems whose exact input--output linearising law is singular on a nonempty set. 
For a family of $p$ laws whose singular hypersurfaces are in general position, the uncovered set is a finite union of embedded submanifolds of codimension exactly $p$: $n+1$ laws always cover an $n$-dimensional state space completely, two laws already leave only a codimension-two residual that generic trajectories never meet, and the number of components of the singularity manifold determines the number of residual pieces rather than their codimension.

The ball-and-beam system provided a complete case study. Its singularity manifold $\{r = 0\}\cup\{\dot\theta = 0\}$ is a regular two-component manifold in general position. Every periodic tracking task forces the state through the component $\{\dot\theta = 0\}$, so the exact law can never be used alone. The approximate law obtained by neglecting the centrifugal term covers the entire admissible region $|\theta|<\pi/2$. Any two of the exact law and the two Hauser--Sastry--Kokotovi\'{c} approximations cover $\R^4$ up to a codimension-two set, and the three laws together cover $\R^4$ minus an explicit obstructed set on which the ball motion is invariant under reversal of the input, so that no linearising law can exist there. Numerical simulations confirmed each of these statements.

Future directions include MIMO extensions, the design and stability analysis of the switching logic in the light of the internal dynamics at the switching surface, the identification of obstructed sets for other mechanical systems with singular relative degree and computational synthesis algorithms for families of approximate laws with prescribed residual codimension.

\section*{Funding}
This research received no external funding.

\section*{Data Availability}
Simulation data and code are available from the author upon reasonable request.

%======================================================================
\bibliographystyle{plainnat}
\bibliography{references}

\end{document}